\documentclass{article}
\usepackage{amsmath,amssymb,amsfonts}

\def\op#1{\mathop{{\it\fam0} #1}\limits}

\newcommand{\beq}{\begin{equation}}
\newcommand{\eeq}{\end{equation}}
\newcommand{\ben}{\begin{eqnarray}}
\newcommand{\een}{\end{eqnarray}}
\newcommand{\be}{\begin{eqnarray*}}
\newcommand{\ee}{\end{eqnarray*}}
\newcommand{\bea}{\begin{eqalph}}
\newcommand{\eea}{\end{eqalph}}

\newcommand{\cA}{{\mathcal A}}

\newcommand{\gA}{{\mathfrak A}}

\newcommand{\cR}{{\mathcal R}}

\newcommand{\si}{\sigma}

\newcommand{\w}{\wedge}

\newenvironment{eqalph}{\stepcounter{equation}
\setcounter{equationa}{\value{equation}} \setcounter{equation}{0}

\begin{eqnarray}}{\end{eqnarray}\setcounter{equation}{\value{equationa}}}

\newcounter{equationa}
\newcounter{remark}
\newcounter{example}
\newcounter{theorem}
\newcounter{proposition}
\newcounter{lemma}
\newcounter{corollary}
\newcounter{definition}
\def\theremark{\arabic{remark}}

\def\thedefinition{\arabic{theorem}}

\newenvironment{proof}{{\it Proof.}}{
\medskip }

\newenvironment{theorem}{\refstepcounter{theorem} \medskip{\bf
Theorem \thedefinition.}\it}{\medskip }

\newcommand{\mar}[1]{}

\begin{document}

\hbox{}

\begin{center}

{\Large\bf

Remark on the Serre -- Swan theorem for graded manifolds}

\bigskip

G. SARDANASHVILY

\medskip

Department of Theoretical Physics, Moscow State University, 117234
Moscow, Russia

\bigskip

\end{center}

\begin{abstract}
Combining the Batchelor theorem and the Serre -- Swan theorem, we
come to that, given a smooth manifold $X$, a graded commutative
$C^\infty(X)$-algebra $\cA$ is isomorphic to the structure ring of
a graded manifold with a body $X$ iff it is the exterior algebra
of some projective $C^\infty(X)$-module of finite rank. In
particular, it follows that odd fields in field theory on a smooth
manifold $X$ can be represented by graded functions on some graded
manifold with body $X$.
\end{abstract}

\bigskip
\bigskip
\bigskip

In classical field theory, there are different descriptions of odd
fields on graded manifolds \cite{cari03,book09,mont06,ijgmmp13}
and supermanifolds \cite{cia95,franc}. Both graded manifolds and
supermanifolds are phrased in terms of sheaves of graded
commutative algebras \cite{bart,book09,sard09}. However, graded
manifolds are characterized by sheaves on smooth manifolds, while
supermanifolds are constructed by gluing sheaves on supervector
spaces. Treating odd fields on a smooth manifold $X$, one can
follow forthcoming Theorem \ref{vv0}. It states that, if a graded
commutative algebra is an exterior algebra of some projective
$C^\infty(X)$-module of finite rank, it is isomorphic to the
algebra of graded functions on a graded manifold whose body is
$X$. By virtue of this theorem, odd fields can be represented by
generating elements of the structure ring of a graded manifold
whose body is $X$ \cite{cmp04,book09,ijgmmp13}.

Let $X$ be a smooth manifold which is assumed to be real,
finite-dimensional, Hausdorff, second-countable, and connected.
The well-known Serre -- Swan theorem establishes the following.

\begin{theorem} \label{ss} \mar{ss}
A $C^\infty(X)$-module $P$ is isomorphic to the structure module
of sections of a smooth vector bundle over $X$ iff it is a
projective module of finite rank.
\end{theorem}

Originally proved for smooth bundles over a compact base $X$, this
theorem has been extended to an arbitrary $X$
\cite{book09,ren,sard01}

Turn now to graded manifolds \cite{bart,book09,sard09}. A graded
manifold of dimension $(n,m)$ is defined as a local-ringed space
$(X,\gA)$ where $X$ is an $n$-dimensional smooth manifold $X$ and
$\gA=\gA_0\oplus\gA_1$ is a sheaf of graded commutative algebras
of rank $m$ such that:

$\bullet$ there is the exact sequence of sheaves
\be
0\to \cR \to\gA \op\to^\si C^\infty_ZX\to 0, \qquad
\cR=\gA_1+(\gA_1)^2,
\ee
where $C^\infty_X$ is the sheaf of smooth real functions on $X$;

$\bullet$ $\cR/\cR^2$ is a locally free sheaf of
$C^\infty_X$-modules of finite rank (with respect to pointwise
operations), and the sheaf $\gA$ is locally isomorphic to the
exterior product $\w_{C^\infty_X}(\cR/\cR^2)$.

A sheaf $\gA$ is called the structure sheaf of a graded manifold
$(X,\gA)$, and a manifold $X$ is said to be the  body of
$(X,\gA)$. Sections of the sheaf $\gA$ are called graded functions
on a graded manifold $(X,\gA)$. They make up a graded commutative
$C^\infty(X)$-ring $\gA(X)$ called the structure ring of
$(X,\gA)$.

The above mentioned Batchelor theorem states the following
\cite{bart,book09}.

\begin{theorem} \label{lmp1a} \mar{lmp1a}
Let $(X,\gA)$ be a graded manifold. There exists a vector bundle
$E\to X$ with an $m$-dimensional typical fibre $V$ such that the
structure sheaf $\gA$ of $(X,\gA)$ is isomorphic to the structure
sheaf $\gA_E=S_{\w E^*}$ of germs of sections of the exterior
product
\mar{ss12f11}\beq
\w E^*=(X\times\mathbb R) \op\oplus_X E^*\op\oplus_X \w^2
E^*\op\oplus_X\cdots\op\oplus_X\w^m E^* \label{ss12f11}
\eeq
of the dual $E^*$ of $E$. Its typical fibre is the Grassmann
algebra
\be
\w V^*=\mathbb R \oplus V\oplus \op\w^2 V\oplus\cdots\oplus\op\w^m
V.
\ee
\end{theorem}

In particular, it follows that the structure ring $\gA(X)$ of a
graded manifold $(X,\gA)$ is isomorphic to the ring of sections of
the exterior product (\ref{ss12f11}).

It should be emphasized that Batchelor's isomorphism in Theorem
\ref{lmp1a} fails to be canonical, but in field models it usually
is fixed from the beginning. We agree to call a graded manifold
$(Z,\gA_E)$ whose structure sheaf is the sheaf of germs of
sections of some exterior bundle $\w E^*$ the simple graded
manifold modelled over a vector bundle $E\to ZX$,

Combining Batchelor Theorem \ref{lmp1a} and Serre -- Swan Theorem
\ref{ss}, we come to the following Serre -- Swan theorem for
graded manifolds \cite{jmp05a,book09,ijgmmp13}.

\begin{theorem} \label{vv0} \mar{vv0}
Let $X$ be a smooth manifold. A graded commutative
$C^\infty(X)$-algebra $\cA$ is isomorphic to the structure ring of
a graded manifold with a body $X$ iff it is the exterior algebra
of some projective $C^\infty(X)$-module of finite rank.
\end{theorem}

\begin{proof}  By virtue of Batchelor Theorem \ref{lmp1a},
any graded manifold is isomorphic to a graded manifold $(X,\gA_E)$
modelled over some vector bundle $E\to X$. Its structure ring
$\cA_E$ of graded functions consists of sections of the exterior
bundle $\w E^*$ (\ref{ss12f11}) and, thus, it is generated by a
$C^\infty(X)$-module $E^*(X)$ of sections of $E^*\to X$. By virtue
of Serre -- Swan Theorem \ref{ss}, this module is a projective
module of finite rank. Conversely, let a graded commutative
$C^\infty(X)$-algebra $\cA$ be generated by some projective
$C^\infty(X)$-module of finite rank. In accordance with the Serre
-- Swan Theorem \ref{ss}, this module is isomorphic to a module of
sections of some vector bundle $E\to X$ and, thus, $\cA$ is
isomorphic to the structure ring of a simple graded manifold
modelled over $E$.
\end{proof}

As a consequence, a graded commutative algebra in Theorem
\ref{vv0} possesses a number of particular properties. For
instance, the Chevalley -- Eilenberg differential calculus of such
an algebra is minimal, and the cohomology of its de Rham complex
equals the de Rham cohomology of a manifold $X$ \cite{book09}.

As physical outcome, let us mention that higher-stage Noether
identities and gauge symmetries of a reducible degenerate
Lagrangian system on a fiber bundle over $X$ are parameterized by
odd fields, called the antifields and ghosts, respectively
\cite{jmp05a,book09}.


\begin{thebibliography}{ederf}

\bibitem{bart} Bartocci, C., Bruzzo, U. and Hern\'andez
Ruip\'erez, D.(1991). \emph{The Geometry of Supermanifolds}
(Kluwer, Dordrecht).

\bibitem{jmp05a} Bashkirov, D., Giachetta, G., Mangiarotti, L. and
Sardanashvily, G. (2005). The antifield Koszul--Tate complex of
reducible Noether identities, \emph{J. Math. Phys.} \textbf{46},
103513; \emph{arXiv}: math-ph/0506034.

\bibitem{cari03} Cari\~nena, J. and Figueroa, H. (2003). Singular Lagrangian
in supermechanics, \emph{Diff. J. Geom. Appl.} \textbf{18}, 33.

\bibitem{cia95} Cianci, R., Francaviglia, M. and Volovich, I.
(1995). Variational calculus and Poincar\'e--Cartan formalism in
supermanifolds, \emph{J. Phys. A.} \textbf{28}, 723.

\bibitem{franc} Franco, D. and Polito, C. (2004). Supersymmetric
field-theoretic models on a supermanifold, \emph{J. Math. Phys.}
\textbf{45}, 1447.

\bibitem{cmp04} Giachetta, G., Mangiarotti, L. and Sardanashvily, G.
(2005). Lagrangian supersymmetries depending on derivatives.
Global analysis and cohomology, \emph{Commun. Math. Phys.},
\textbf{259}, 103.

\bibitem{book09} Giachetta, G., Mangiarotti, L. and Sardanashvily,
G. (2009) \emph{Advanced Classical Field Theory} (World
Scientific, Singapore).

\bibitem{mont06} Monterde, J., Masqu\'e, J. and Vallejo, J. (2006). The
Poincar\'e--Cartan form in superfield theory, \emph{Int. J. Geom.
Methods Mod. Phys.} \textbf{3}, 775.

\bibitem{ren} Rennie, A. (2003). Smoothness and locality for
nonunital spectral triples, \emph{K-Theory} \textbf{28}, 127.

\bibitem{sard01} Sardanashvily, G. (2013). Remark on the Serre-Swan theorem for
non-compact manifolds, \emph{arXiv}: math-ph/0102016.

\bibitem{sard09} Sardanashvily, G. (2013). Lectures on
supergeometry, \emph{arXiv}: 0910.0092.

\bibitem{ijgmmp13} Sardanashvily, G. (2013). Graded Lagrangian
formalism, \emph{Int. J. Geom. Methods Mod. Phys.} \textbf{10}, N5
1350016; \emph{arXiv}: 1206.2508.




\end{thebibliography}
\end{document}